\documentclass{llncs}
\usepackage{multirow}
\usepackage{todonotes}
\usepackage{thmtools}
\usepackage{thm-restate}
\usepackage[capitalise]{cleveref}
\usepackage{graphics,adjustbox}
\usepackage{tabularx}
\usepackage{amssymb}
\usepackage{epsfig}
\usepackage{verbatim}
\usepackage{times}
\usepackage{pstricks}
\usepackage{pstricks-add}
\usepackage{pst-text}
\usepackage{fancybox}
\usepackage{pst-node}
\usepackage{graphics}
\usepackage{tikz-qtree}
\usepackage{forest}
\usepackage{algpseudocode,algorithm}
\usepackage{float}
\usepackage{wrapfig}
\usepackage{mathtools}
\usepackage{enumitem}
\usepackage{cite}
\usepackage{graphicx}
\usepackage{sidecap}

\usetikzlibrary{arrows}
\usetikzlibrary{positioning}
\usepackage{tikz}
\usetikzlibrary{trees}
\usetikzlibrary{decorations.pathmorphing}
\usetikzlibrary{decorations.markings}
\usetikzlibrary{decorations.pathmorphing,shapes}
\usetikzlibrary{calc,decorations.pathmorphing,shapes}

\usepackage{tikz}
\usepackage{verbatim}
\usetikzlibrary{trees, arrows, shapes}
\tikzset{
  treenode/.style = {align=center, inner sep=2pt, text centered,
    font=\sffamily},
  arn_r/.style = {treenode, circle, black, font=\sffamily\bfseries, draw=black,
    text width=1.5em},
    arn_t/.style = {treenode, circle, black, thick, double, font=\sffamily\bfseries, draw=black,
    text width=1.5em},
  every edge/.append style={anchor=south,auto=falseanchor=south,auto=false,font=3.5 em},
}

\usetikzlibrary{positioning}
\usetikzlibrary{arrows}
\usetikzlibrary{shapes.multipart}

\def\dd{\mathinner{.\,.}}
\newcommand{\cO}{\mathcal{O}}

 \newcommand{\defproblem}[3]{
  \vspace{2mm}
\noindent\fbox{
  \begin{minipage}{0.96\textwidth}
  #1\\
  {\bf{Input:}} #2  \\ 
3

  {\bf{Output:}} #3
  \end{minipage}
  }
  \vspace{2mm}
}

\usepackage[super]{nth}
\usepackage[T1]{fontenc}
\usepackage{lmodern}
\usepackage{setspace}

\begin{document}

\title{On the cyclic regularities of strings}
%
%
\author{ Oluwole Ajala\inst{1}, Miznah Alshammary\inst{1}, Mai Alzamel\inst{1}, Jia Gao\inst{1},  Costas Iliopoulos\inst{1}, Jakub Radoszewski\inst{2}, Wojciech Rytter\inst{2} and Bruce Watson\inst{3}  }

\institute{
Faculty of Natural and Mathematical Sciences, King's College London, United Kingdom\\
\email{\{oluwole.ajala, miznah.alshammary, mai.alzamel, jia.gao, c.iliopoulos\}@kcl.ac.uk}
\and
Faculty of Mathematics, Informatics and Mechanics, University of Warsaw, Poland\\
\email{\{jrad, rytter\}@mimuw.edu.pl}
\and
Faculty of Informatics Science, Stellenbosch University, South Africa\\
\email{\{bwwatson\}@sun.ac.za}
}
\maketitle              
%
\begin{abstract}
Regularities in strings are often related to periods and covers, which have extensively been studied, and algorithms for their efficient computation have broad application. In this paper we concentrate on computing cyclic regularities of strings, in particular, we propose several efficient algorithms for computing:
(i) cyclic periodicity; (ii) all cyclic periodicity; (iii) maximal local cyclic periodicity; (iv) cyclic covers. 
\end{abstract}
%
%
\section{Introduction and Related Work}

A fundamental concept of repeating patterns or {\em regularities} is that of periods (also known as powers). A period of order $k$ is  defined by a concatenation of $k$ identical blocks of symbols.
The study of periods can be traced to as far back as the early 1900s with the work of~\cite{thue1906uber}, who researched a set of strings that do not contain any substrings that are periods. Periods in diverse forms gained prominence, when they became key structures in computational biology, where they are associated with various regulatory mechanisms and play an important role in genomic fingerprinting~\cite{kolpakov2003mreps}.

In regularities in strings, one of the most general notions is related to period or power, for instance, given a string $x$ of length $n$, a period $k$ 
of a string $x$ is a sub string of $x$, if it can be decomposed into 
equal-length blocks of symbols, such that $x$= $u^k u'$, where $u'$ is a prefix of $u$. However, for simplicity we will discard $u'$ and only consider $u^k$.


So far, regularities in strings related with periods and powers, which have been extensively studied,~\cite{defant2016anti}~\cite{erdHos1973anti}~\cite{fujita2010rainbow}~\cite{narayanan2017functions} and algorithms for their efficient computation have broad applications. In this paper, we study cyclic factors of strings. The motivation of cyclic factors comes from viruses. The viruses are circular strings, for example Escherichia coli (E.coli) has 154 bases and it is circular~\cite{wurpel2013chaperone}(Fig. 1). Formally, the viruses break up at any point of the circle, for example, that can appear in the DNA sequence as $x_\delta \dots x_n x_1 \dots x_{\delta-1}$ breaking up at position $\delta$ (Fig.~\ref{fig:cyclic1}).
Now, we propose several efficient algorithms for computing:
(i) cyclic periodicity; (ii) all cyclic periodicity; (iii) maximal local cyclic periodicity; (iv) cyclic covers. 
\begin{SCfigure}
  \centering
  \includegraphics[width=0.48\textwidth]%
    {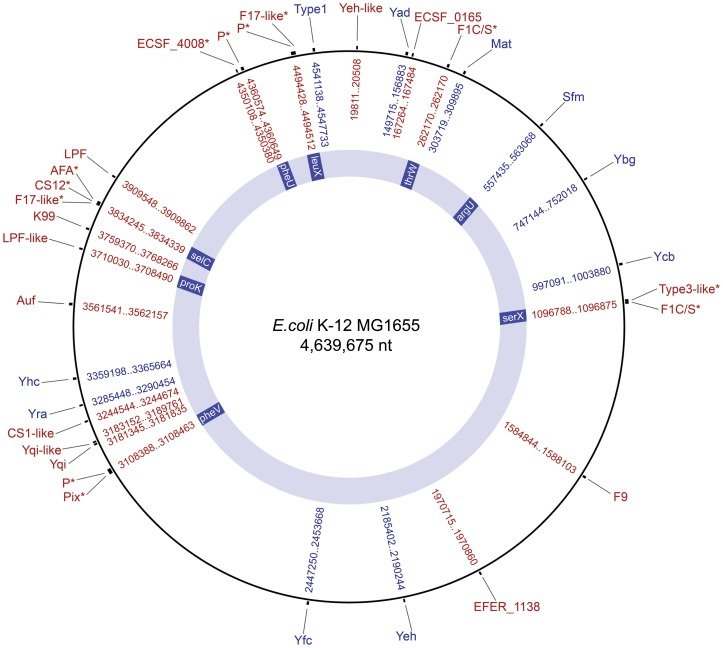}
      \caption{The E. coli K-12 MG1655 chromosome (outer black ring) was used as a reference map to visualise the locus position of 30 chromosome-borne CU fimbrial types. Types highlighted in blue are present in E. coli K-12 MG1655, types in red are absent in this strain. Fimbrial types associated with PAIs are indicated by an asterisk. A number of PAI associated fimbrial gene clusters occupy different locus positions relative to the MG1655 genome. tRNA sites that flank CU-containing PAIs are indicated on the inner blue ring.~\cite{wurpel2013chaperone}}
\end{SCfigure}


\begin{figure}[h]
	\centering
	\includegraphics[width=0.3\textwidth]{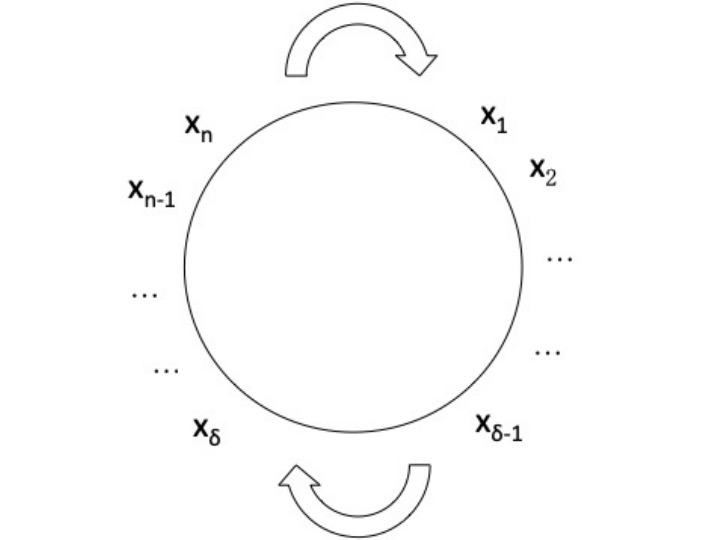}
	\caption{Circular Pattern.}
	\label{fig:cyclic1}
\end{figure}

\section{Preliminaries}

A \textit{string} $x$ of length $|x|=n$ over an alphabet $\Sigma$ of size $\sigma$ can be denoted as $x[1 \dd n]=x[1]x[2] \dots x[i] \dots x[n]$, where $1 \leq i \leq n$ and the $i$-$th$ \textit{letter} of $x$ is denoted by $x[i]$ $\in$ $\Sigma$. 
The empty string $\epsilon$ is the string of length 0. The string $x^R$ is the \textit{reverse} of string $x$.
And $x[i \dd j]$, $1 \leq i \leq j \leq n$, denotes the contiguous substring (or factor) of letters, such as $x[i]x[i+1]x[i+2] \dots x[j]$.
A substring $x[i \dd j]$ is a suffix of $x$, if $j=n$ and a substring $x[i \dd j]$ is a prefix of $x$, if $i=1$.
Given a cyclic factor $u$ of length $k$, $1 \leq k \leq n$, we denote by $c(u)$, for instance, $u=ababc$, $c(u)$ is one of the following rotations: $ababc$, $babca$, $abcab$, $bcaba$, $cabab$.
Moreover, we say $u=u_1 u_2 \dots u_n$, $c_{\delta}(u)=u_{\delta} \dots u_n u_1 \dots u_{\delta -1}$.

In this paper, suffix trees are used extensively as  computation tools. For a general introduction to suffix trees, see~\cite{CHL07},\cite{mccreight1976space},\cite{ukkonen1992constructing}, \cite{weiner1973linear}.

{\defproblem{\textsc{K-Cyclic Period}}{Given a string $x$ of length $n$, and an integer $k < n$, compute k-cyclic-period $\ell$ of $x$, where $x ={u_1}{u_2} \dots {u_\ell}$, ${u_i}={c}(u_j)$, $|{u_i}|=|{u_j}|=k$, $\forall$ $i,j$, $1 \leq i \leq \ell$, $1 \leq j \leq \ell$, and $k\times\ell=n$.}{k-cyclic-period $\ell$ of $x$}}

\begin{example} 
Consider a string $x=aaabaabaabaabaaa=:{u_1}{u_2}{u_3}{u_4}$ , where $u_1=aaab$, $u_2=aaba$, $u_3=abaa$, $u_4=baaa$ and $k=4$, $\ell=4$. Therefore $x$ has a period of length $\ell$.
\end{example}

\begin{definition}
A cyclic periodic array $A$ of a string $x$ of length $n$ is defined to be as follows: $A[i]:=\ell$, $1 \leq i \leq n$, if and only if $x[1 \dd i]$ has cyclic periodicity $\ell$ by a string $u$ and there no $u'$ , with  $|u'|\leq |u|$ that is a cyclic period of $x[1 \dd i]$.
\end{definition}

\begin{example} 
Consider a string $x=aababa$ of length $6$, a cyclic periodic array $A$ as follows:

$x[1]=a \implies A[1]:=1$ \quad \quad \quad \quad \quad $ $ $x[1 \dd 4]=aaba \implies A[4]:=1$

$x[1 \dd 2]=a$ $a \implies A[2]:=2$ \quad \quad \quad $x[1 \dd 5]=aabab \implies A[5]:=1$

$x[1 \dd 3]=aab \implies A[3]:=1$ \quad \quad \quad $x[1 \dd 6]=aab$ $aba \implies A[6]:=2$
\end{example}

\begin{definition}
We define maximal local $k$-cyclic periodicity of a string $x$ ,if a substring $y$ is cyclic periodic
and $y$ is not a substring of another cyclic periodic strings. 
\end{definition}

\begin{example}
Consider a string $x=aaaababaaab$, $\Sigma=\{a,b\}$, and a substring $y=aabababaa$ is $3$-cyclic periodic and substring $y\alpha=aabababaaa$, $\alpha \in \Sigma$,  is not cyclic periodic, and substring $\beta y=aaabababaa$, $\beta \in \Sigma$ is not cyclic periodic. Therefore, the substring $y=aabababaa$ is maximal local $3$-cyclic periodic in string $x=aaaabababaaa$.
\end{example}
\begin{definition}
We say that a string $x$ of length $n$ is cyclic-coverable by a string $u$ of length $k'$, if and only if, for every position $i$
of $x$, the following condition holds $x[\beta \dd \gamma]=c(u)$, $1 \leq \beta \leq i \leq \gamma \leq n$.
\end{definition}

\begin{example}
Consider a string $x=aababaa={u_1}{u_2}$, $u_1=aaba$, $u_2=abaa$, $k'=4$, $\gamma=2$, is cyclic coverable by a string $u$, for every position $i$ of $x$, $x[1 \dd 4]=x[4 \dd 7]=c(u)$.
\end{example}

\begin{definition}
Compute all cyclic covers of a given string $x$, that is for all possible length cyclic covers.
\end{definition}

\begin{example}
Consider a string $x=ababbaba$, then $ab$, $abab$, $ababb$, $ababbab$ are cyclic covers of $x$.
\end{example}





\section{Computing $k$-cyclic periodicity}

\begin{theorem}
Given a string $x$ of length $n$ and an integer $k$, $1 \leq k \leq n$, test whether it is $k$-cyclic periodic; this can be determined in $\cO(n/k)$ time and $\cO(n)$ space.
\end{theorem}
\begin{proof}
We construct the suffix tree of $x$ (see~\cite{mccreight1976space},~\cite{ukkonen1992constructing},~\cite{weiner1973linear}).
We let $u=x[1\dots k]$, then let $\ell_m$ denote the depth of the lowest common ancestor of $x[1 \dd n]$ (see~\cite{Bender2000}), and $x[i_m \dd n]$. We compute the LCA $\ell_m$ of $x[1 \dd n]$ and $x[i_m \dd n]$ for $i_m=2k, 3k, \dots$, and $\ell_k=n-k$, if $\ell_m=1$ for some $m$, then $x$ is not $k$-cyclic periodic string. Now consider $C^{rignt}_m=(u_{\ell_m+1} \dots u_k)^R$, compute the $\ell_m^{'}$ the LCA of $u^R$ and $C^{right}_m$. If $\ell_m^{'} \geq \ell_m$ for all $m$, then $x$ is $k$-cyclic periodic. \qed
\end{proof}


\section{Computing all cyclic periodicities}

\begin{theorem}
Given a string $x$ of length $n$, test whether it is  $k$-cyclic periodicity for all $1 \leq k \leq n$, this can be determined in $\cO(n \log n)$ time and $\cO(n)$ space.
\end{theorem}
\begin{proof}
We apply the algorithm of Theorem 1 for $k=1,2,\dots,n$ and we test all cyclic periods of length $k$. The construction of the suffix tree of string $x$ and $x^R$ is done once costing $\cO(n)$. The total cost is $$\mathcal{O}(n)+\cO(\sum_{k=1}^{n} n/k)=\mathcal{O}(n\log n)$$. \qed
\end{proof} 
\begin{lemma}
Compute the cyclic period of $x$. 
\end{lemma}.
\begin{proof}
The smallest cyclic-period of $x$ is the cyclic-period of $x$. \qed
\end{proof}
\section{Computing maximal local $k$-cyclic periodicity}

\begin{theorem}
We can compute all $k$-cyclic periodicity of $x$ in $\cO(n \log n)$ time.
\end{theorem}
\begin{proof}
We apply the algorithm for $k=1,2,\dots,n$ and in this case, extend it to cyclic periods of length $k+1$, where $|y|$ = $m$ is cyclic periodic and $y\alpha$ = $m+1$ is not cyclic periodic. Next, we perform this algorithm on string $x^R$ as $\mathcal{T}(x^R)$, where $|y|$ = $m$, again is cyclic periodic and $\beta y$ = $m+1$ is not cyclic periodic. 

The construction of the suffix tree of string $x$ is done once. The total cost is 
$$\cO(\sum_{k=1}^{n} n/k)=\cO(n\log n)$$. \qed
\end{proof}
\begin{lemma}
Compute maximal local $k$-cyclic periodicity of $x$. 
\end{lemma}.
\begin{proof}
We compute and merge the arrays for $y\alpha$ and $\beta y$ of $x$. That is the maximal local $k$-cyclic periodicity of $x$. \qed
\end{proof}

\section{Computing $k'$-cyclic coverability}

\begin{theorem}
Given a string $x$ of length $n$ and an integer $k'$, $1 \leq k' \leq n$, test whether it is $k'$-cyclic coverable, this can be determined in $\cO(n)$ time and $\cO(n)$ space.
\end{theorem}
\begin{proof}
We compute the suffix tree of string $x$ as $\mathcal{T}(x)$, and also we compute the suffix tree of string $x^R$ as $\mathcal{T}(x^R)$.

Then we check $x[1,k']$ with each one of $x[n-k'+1,n]$, $x[n-k',n-1]$, $x[n-k'-1,n-2]$ \dots $x[2,k'+1]$, together with the reverse pairs in $T(x^R)$. This way we build a collection of cyclic covers if there is one.




The construction of the suffix tree costs $\cO(n)$; checking of equality costs $\cO(1)$ and there are $n$ factors. The total time is $\cO(n)$. \qed
\end{proof}


\section{Computing all cyclic coverability}

\begin{theorem}
Given a string $x$ of length $n$, test whether it is  $k'$-cyclic coverable for $1 \leq k' \leq n$, this can be determined in $\cO(n^2)$ time and $\cO(n)$ space.
\end{theorem}
\begin{proof}
We apply the algorithm for $k'=1,2,\dots,n$ and we compare all cyclic coverable of length $k'$. The construction of the suffix tree of string $x$ is done once. The total cost is 
$$\cO(\sum_{k'=1}^{n} n)=\cO(n^2)$$. \qed
\end{proof}
\begin{lemma}
Compute the cyclic coverability of $x$. 
\end{lemma}.
\begin{proof}
The smallest cyclic coverable of $x$ is the all the cyclic coveralbe of $x$. \qed
\end{proof}









\section{Conclusions and open problems}
In this paper, we defined $k$-cyclic periodicity,
we presented several efficient algorithms for computing: (i) cyclic periodicity; (ii) all cyclic periodicity; (iii) maximal local cyclic periodicity; (iv) cyclic covers.

Future work will be focused on computing the cyclic-periodic array, that is the cyclic periodicity of every prefix of string $x$ and  computing the cyclic-coverability array, that is testing each prefix of $x$, for cyclic-coverability. Finally, we will extend the cyclic periodicity to cover the case $u_1u_2u_2\dots u_k u^1$, where $u_i$=$c(u_j)$ $\forall$ $i,j$ and $u^1$ is a substring of some $u_i$.

%
%
%
%
\bibliographystyle{splncs03}
\bibliography{main.bib}

\end{document}